\documentclass{article}
\usepackage{amsmath}
\usepackage{amsthm}
\usepackage{amssymb}
\usepackage{amsfonts}
\usepackage{epic}
\usepackage{eepic}
 \usepackage{times}
\usepackage[matrix,arrow]{xy}
\usepackage{macros}
\usepackage{epsfig}
\theoremstyle{plain}

\newtheorem{thm}{Theorem}

\newtheorem{lem}{Lemma}

\newtheorem{defn}{Definition}

\theoremstyle{remark}

\newcommand{\sst}{\scriptscriptstyle}

\renewcommand{\1}{\one}
\renewcommand{\2}{\two}

\newcommand{\beq}{\begin{equation}}
\newcommand{\eeq}{\end{equation}}

\newcommand{\ot}{\otimes}
\newcommand{\ra}{\to}

\newcommand{\SRN}{{\rm N}}

\newcommand{\de}{\delta}

\newcommand{\la}{\lambda}

\newcommand{\CA}{{\mathcal A}}
\newcommand{\CB}{{\mathcal B}}

\newcommand{\CD}{{\mathcal D}}

\newcommand{\CQ}{{\mathcal Q}}
\newcommand{\CR}{{\mathcal R}}

\newcommand{\SA}{{\mathsf A}}
\newcommand{\SB}{{\mathsf B}}
\newcommand{\SC}{{\mathsf C}}
\newcommand{\SD}{{\mathsf D}}

\newcommand{\SG}{{\mathsf G}}

\newcommand{\SM}{{\mathsf M}}

\newcommand{\SO}{{\mathsf O}}

\newcommand{\SQ}{{\mathsf Q}}

\newcommand{\ST}{{\mathsf T}}
\newcommand{\SU}{{\mathsf U}}

\newcommand{\SW}{{\mathsf W}}

\newcommand{\su}{{\mathsf u}}
\newcommand{\sv}{{\mathsf v}}

\newcommand{\one}{{\mathfrak 1}}
\newcommand{\two}{{\mathfrak 2}}

\newcommand{\BC}{{\mathbb C}}

\newcommand{\BS}{{\mathbb S}}

\newcommand{\BZ}{{\mathbb Z}}

\newcommand{\srn}{{\sst\rm N}}
\newcommand{\SRM}{{\rm M}}

\newcommand{\rf}[1]{(\ref{#1})}

\newcommand{\en}{{\rm e}_\SRN}

\newcommand{\aufz}
{\begin{list}{$\bullet$}{\topsep0cm \itemsep0cm \parsep0cm}}
\newcommand{\eaufz}{\end{list}}

\begin{document}


\title{Reconstruction of Baxter $\SQ$-operator from Sklyanin SOV \newline for cyclic representations of integrable quantum models}

\author{G. Niccoli}

\address{DESY, Notkestr. 85, 22603 Hamburg, Germany}

\maketitle

{\vspace{-7cm} \tt {DESY 09-227}}

\vspace{6cm}

{\bf Abstract}
\\
\begin{quotation}
\hspace{-0.7cm}  \ In \cite{NT}, the spectrum (eigenvalues and eigenstates) of a lattice regularizations of the Sine-Gordon model
has been completely characterized in terms of polynomial solutions with certain properties of the
Baxter equation. This characterization for cyclic representations has been derived by the use of the Separation of
Variables (SOV) method of Sklyanin and by the direct construction of the Baxter $\SQ$-operator family. Here, we
reconstruct the Baxter $\SQ$-operator and the same characterization of the spectrum by only using
the SOV method. This analysis allows us to deduce the main features required for the extension to cyclic representations of other integrable quantum models of this kind of spectrum characterization.    

\end{quotation}

\vspace{4cm}

\vspace{1cm}
\newpage

\section{Introduction}
The integrability of a quantum model is by definition related to the existence of a mutually commutative family $\CQ$ of self-adjoint operators $\ST$ such that
\begin{equation}\label{def-integrability}
\begin{aligned}
& {\rm (A)}\quad[\,\ST\,,\,\ST'\,]\,=\,0,\\
& {\rm (B)}\quad[\,\ST\,,\,\SU\,]\,=\,0,\\
& {\rm (C)}\quad{\rm if}\;\;[\,\ST\,,\,\SO\,]\,=\,0\,,
\end{aligned}\;\;
\begin{aligned}
& \forall\, \ST,\ST'\in\CQ\,,\\
& \forall\,\ST\in\CQ\,,\\
&\forall \,\ST \in \CQ, \;\;{\rm then}\;\;\SO=\SO(\CQ)\,,
\end{aligned}\end{equation}
where $\SU$ is the unitary operator defining the time-evolution in the model; note that the property (C) stays for the completeness of the family $\CQ$. In the framework of the quantum inverse scattering method \cite{FT79, KS79, FST} the Lax operator $L(\la)$ is the mathematical tool which allows to define the transfer matrix:
\begin{equation}\label{Tdef}
\ST^{}(\la)\,=\,{\rm tr}_{\BC^2}^{}\SM(\la)\,, \qquad
\SM(\la)\,\equiv\,\left(\begin{matrix}\SA(\la) & \SB(\la)\\
\SC(\la) & \SD(\la)\end{matrix}\right)\,\equiv\,
L_\SRN^{}(\la)\dots L_1^{}(\la)\,,
\end{equation}
a one parameter family of mutual commutative self-adjoint operators. The integrability of the model follows from  $\ST(\la)$ if the properties (B) and (C) of definition \rf{def-integrability} can be proven for it. In some quantum model the integrability is derived by proving the existence of a further one-parameter family of self-adjoint operators the $\SQ$-operator which by definition satisfies the following properties:
\begin{equation}
[\, \SQ(\la)\,,\,\SQ(\mu)\,]\,=\,0\,, \qquad [\, \ST(\la)\,,\,\SQ(\mu)\,]\,=\,0\,,\qquad \forall \la,\mu\,\in\BC,
\end{equation}
plus the Baxter equation with the transfer matrix:
\begin{equation}\label{BAX}
\ST(\la)\SQ(\la)\,=\,{\tt a}(\la)\SQ(q^{-1}\la)+{\tt d}(\la)\SQ(q\la)\,.
\end{equation}
This is in particular the case for those models (like Sine-Gordon \cite{NT}) for which the time-evolution operator $\SU$ is expressed in terms of $\SQ$. A natural question arises: Is the integrable structure of these quantum models completely characterized by the transfer matrix $\ST(\la)$?

Note that a standard procedure\footnote{It is worth recalling that
there are also others constructions of the $\SQ$-operator. An interesting
example is presented in the series of works \cite{BLZ-I, BLZ-II,
BLZ-III} by V.V. Bazhanov, S.L. Lukyanov and A.B. Zamolodchikov on the
integrable structure of conformal field theories. In \cite{BLZ-II, BLZ-III} the $\SQ$-operator
is obtained as a transfer-matrix by a trace procedure of a fundamental
$L$-operator with $q$-oscillator representation for the auxiliary space (see
also \cite{AF96, RW02}). This construction can be extended to massive
integrable quantum field theories as it was argued by the same authors
in \cite{BLZ-IV}.} to prove the existence of $\SQ(\la)$ is by a direct
construction of an operator solution of the Baxter equation \rf{BAX}. Moreover, the coefficients ${\tt a}(\la)$ and ${\tt d}(\la)$ as well as the
analytic and asymptotics properties of $\SQ(\la)$ are some model dependent features which are derived by the
construction.
Let us recall that the general strategy \cite{Ba73, BS, GP, De, DKM} of this construction is to find a {\it gauge}
transformation\footnote{It leaves unchanged the transfer matrix while modifies the monodromy matrix $\SM(\la)$ defined in \rf{Tdef} .} such that the action of each gauge transformed Lax
matrix on $\SQ(\la)$ becomes upper-triangular. Then the $\SQ$-operator assumes a factorized {\it local} form and
the problem of its existence in such a form is reduced to the problem of the existence of some model dependent {\it special} function\footnote{The {\it quantum dilogarithm} functions \cite{FK2, F2, Ru, Wo, PT2, K1, K2, BT03, T2, V2} for example appear in the Sinh-Gordon model \cite{BT06}, in their {\it non-compact} form, and in the Sine-Gordon model \cite{NT}, in their {\it cyclic} form.}.

It is worth pointing out that on the one hand the construction of these special functions for general models can
represent a concrete technical problem\footnote{The Sine-Gordon model at irrational values of the coupling
$\beta^2$ is a simple case where this kind of problem emerges.} and that on the other hand the existence of
such functions is only a sufficient criterion for the existence of $\SQ(\la)$.
It is then a relevant question if it is possible to bypass this kind of construction providing a different proof of the
existence of $\SQ(\la)$.  

\vspace{0.4cm}

Given an integrable quantum model the first fundamental task to solve is the exact solution of its {\it spectral
problem}, i.e. the determination of the eigenvalues and the simultaneous eigenstates of the operator family $\CQ$, defined in \rf{def-integrability}.
There are several methods to analyze this spectral problem as the {\it coordinate} Bethe ansatz \cite{Be31, Bax82, ABBQ87}, the $\ST\SQ$ method \cite{Bax82}, the {\it algebraic} Bethe ansatz (ABA) \cite{FT79, KS79, FST}, the {\it analytic} Bethe ansatz \cite{R83} and the separation of variables (SOV) method of Sklyanin \cite{Sk1,Sk2,Sk3}; this last one seems to be more promising. Indeed, on the one hand it resolves the problems related to the reduced applicability of other methods (like ABA) and on the other hand it directly implies the completeness of the characterization of the spectrum which instead for other methods has to be proven.

For cyclic representations \cite{Ta} of integrable quantum models the SOV method should lead to the characterization of
the eigenvalues and the simultaneous eigenstates of the transfer matrix $\ST(\la)$ by a {\it finite}\footnote{The
number of equations in the system is finite and related to the dimension of the cyclic representation.} system of
Baxter-like equations. However, it is worth pointing out that such a characterization of the spectrum is not the most
efficient; this is in particular true in view of the analysis of the continuum limit. Here the main question reads: Is it
possible to define a set of conditions under which the SOV characterization of the spectrum can be reformulated
in terms of a functional Baxter equation? In fact, this is equivalent to ask if we can reconstruct the $\SQ$-operator
from the finite system of Baxter-like equations. In this case the solution of the spectral problem is
reduced to the classification of the solutions of the Baxter equation which satisfy some analytic and asymptotic
properties fixed by the operators $\ST$ and $\SQ$.

The lattice Sine-Gordon model is used as a concrete example where these questions about quantum integrability
find a complete and affirmative answer. Indeed, in section 3, we show that the SOV characterization of the
transfer matrix spectrum is exactly equivalent to a functional equation of the form $\det D(\Lambda
)=0$, where $D(\lambda )$  (see \rf{D-matrix}) is a one-parameter family of {\it quasi-tridiagonal} matrices. In
section 4, we show that this functional equation is indeed equivalent to the Baxter functional equation and, in
section 5, we use these results to reconstruct the Baxter $\SQ$-operator with the same level of accuracy obtained
by the direct construction presented in \cite{NT}. It is worth pointing out that these results allow us to prove that
the transfer matrix $\ST^{}(\la)$ (plus the $\Theta$-charge for even chain) describes the family $\CQ$ of
complete commuting self-adjoint charges which implies the quantum integrability of the model according to
definition \rf{def-integrability}. So that in the Sine-Gordon model the Baxter $\SQ$-operator plays only the role of
a useful auxiliary object.

Let us point out that one of the main advantages of the spectrum characterization derived for the Sine-Gordon model is the
possibility to prove an exact reformulation in terms of non-linear integral equations\footnote{This type of equations were before introduced in a different framework in \cite{KP91,KBP91}} (NLIE). This will be the subject of a
future publication where the NLIE characterization will lead us by the implementation of the continuum limit to the
description of the Sine-Gordon spectrum in all the interesting regimes. These results will be shown to be consistent with those obtained previously in the literature\footnote{See \cite{FR02I, FR02II} for a related model analyzed in the framework of ABA and \cite{FR03I, FR03II} for the corresponding finite volume continuum limit.} \cite{DDV92, DDV94, DDV97, FMQR96, FRT98, FRT99} (see \cite{F00, R01} for reviews). Note that the method based on the reformulation of the spectral problem in terms of NLIE has been also used recently \cite{T} to derive the Sinh-Gordon spectrum in finite volume and to characterize the spectrum in the  infrared and ultraviolet limits.

 The analysis of the Sine-Gordon model allows us to infer the main features required to extend this kind of spectrum characterization to cyclic representations of other integrable quantum models. This is particularly relevant for those models for which a direct construction of the Baxter $\SQ$-operator encounters technical difficulties.
\vspace*{1mm}
{\par\small
{\em Acknowledgments.} I would like to thank  J. Teschner for stimulating discussions and suggestions on a preliminary version of this work and J.-M. Maillet for the interest shown.

I gratefully acknowledge support from the EC by the Marie Curie Excellence Grant MEXT-CT-2006-042695.}

\section{The Sine-Gordon model}\label{SOV}
We use this section to recall the main results derived in \cite{NT} on the description in terms of SOV of the lattice Sine-Gordon model. This will be used as the starting point to introduce a characterization of the spectrum of the transfer matrix $\ST(\la)$ which will lead to the construction of the $\SQ$-operator from SOV.

\subsection{Definitions}\label{T-op}
The lattice Sine-Gordon model can be characterized by the following Lax matrix\footnote{The lattice regularization of the Sine-Gordon model that we consider here goes back to \cite{FST,IK} and is related to formulations which have more recently been studied in  \cite{FV94,BBR,Ba08}.}:
\begin{equation}\label{Lax}\begin{aligned}
 L^{\rm\sst SG}_n(\la)
  &= \frac{\kappa_n}{i} \left( \begin{array}{cc}i\,\su_n^{}(q^{-\frac{1}{2}}\kappa_n^{}\sv_n^{}+q^{+\frac{1}{2}}\kappa^{-1}_n\sv_n^{-1}) &
\la_n^{} \sv_n^{} - \la^{-1}_n \sv_n^{-1}  \\
 \la_n^{} \sv_n^{-1} - \la^{-1}_n \sv_n^{} &
i\,\su_n^{-1}(q^{+\frac{1}{2}}\kappa^{-1}_n\sv_n^{}+q^{-\frac{1}{2}}\kappa_n^{}\sv_n^{-1})
 \end{array} \right) ,
\end{aligned}\end{equation}
where $\la_n\equiv\la/\xi_n$ for any $n\in \{1,...,\SRN\}$ with $\xi_n$ and  $\kappa_n$ parameters of the model.
For any $n \in \{1,...,\SRN\}$ the couple of operators ($\su_n$,$\sv_n$) define a Weyl algebra ${\cal W}_{n}$:
\begin{equation}\label{Weyl}
\su_n\sv_m=q^{\de_{nm}}\sv_m\su_n\,,\qquad{\rm where}\;\;
q=e^{-\pi i \beta^2}\,.
\end{equation}
We will restrict our attention to the case in which $q$ is a root of unity,
\begin{equation}\label{beta}
\beta^2\,=\,\frac{p'}{p}\,,\qquad p,p'\in\BZ^{>0}\,,
\end{equation}
with $p \equiv 2l+1$ odd and $p'$ even so that $q^{p}=1$. In this case each Weyl algebra  ${\cal W}_{n}$ admits a finite-dimensional representation of dimension $p$. In fact, we can represent the operators $\su_n$, $\sv_n$ on the space of complex-valued functions $\psi:\BS_p^{\SRN}\ra \BC$ as
\begin{equation}\label{reprdef}
\begin{aligned}
&\su_n\cdot\psi(z_1,\dots,z_\SRN)\,=\,u_n z_n\psi(z_1,\dots,
z_n,\dots,z_\SRN)\,,\\
&\sv_n\cdot\psi(z_1,\dots,z_\SRN)\,=\,v_n\psi(z_1,\dots, q^{-1} z_n,\dots,z_\SRN)\,.
\end{aligned}
\end{equation}
where $\BS_p=\{q^{2n};n=0,\dots,2l\}$ is a subset of the unit circle; note that
$\BS_p=\{q^{n};n=0,\dots,2l\}$ since $q^{2l+2}=q$.

The monodromy matrix $\SM(\la)$ defined in \rf{Tdef} in terms of the Lax-matrix \rf{Lax} satisfies the quadratic relations:
\begin{equation}\label{YBA}
R(\la/\mu)\,(\SM(\la)\ot 1)\,(1\ot\SM(\mu))\,=\,(1\ot\SM(\mu))\,(\SM(\la)\ot 1)R(\la/\mu)\,,
\end{equation}
where the auxiliary $R$-matrix is given by
\begin{equation}\label{Rlsg}
 R(\la) =
 \left( \begin{array}{cccc}
 q^{}\la-q^{-1}\la^{-1} & & & \\ [-1mm]
 & \la-\la^{-1} & q-q^{-1} & \\ [-1mm]
 & q-q^{-1} & \la-\la^{-1} & \\ [-1mm]
 & & &  q\la-q^{-1}\la^{-1}
 \end{array} \right) \,.
\end{equation}

The elements of $\SM(\la)$ generate a representation $\CR_\SRN$ of the so-called Yang-Baxter algebra characterized by the $4\SRN$ parameters $\kappa=(\kappa_1,\dots,\kappa_\SRN)$,
$\xi=(\xi_1,\dots,\xi_\SRN)$, $u=(u_1,\dots,u_\SRN)$ and $v=(v_1,\dots,v_\SRN)$; in the present paper we will restrict to the case $u_n=1$, $v_n=1$, $n=1,\dots,\SRN$. The commutation relations \rf{YBA} are at the basis of the proof of the mutual commutativity of the $\ST$-operators.

In the case of a lattice with $\SRN$ even quantum sites, we have also to introduce the operator:
\begin{equation}  \label{topological-charge}
\Theta =\prod_{n=1}^{\SRN}\sv_{n}^{(-1)^{1+n}},
\end{equation}
which plays the role of a {\it grading operator} in the Yang-Baxter algebra:

{\bf Proposition \ 6 \ of \ \cite{NT}} \ $\Theta $ {\it commutes with the transfer matrix and satisfies
the following commutation relations with the entries of the monodromy matrix:}
\begin{eqnarray}
\Theta \SC(\lambda ) &=&q\SC(\lambda )\Theta \text{, \ \ \ }[\SA(\lambda ),\Theta
]=0, \\
\SB(\lambda )\Theta &=&q\Theta \SB(\lambda ),\text{ \ \ }[\SD(\lambda ),\Theta ]=0.
\end{eqnarray}
Moreover, the $\Theta $-charge allows to express the asymptotics of the transfer matrix as:
\begin{equation}
\lim_{\log\lambda \rightarrow \mp\infty}\lambda ^{\pm\SRN}\ST (\lambda )=\left(
\prod_{a=1}^{\SRN}\frac{\kappa _{a}\xi _{a}^{\pm 1}}{i}\right) \left(
\Theta +\Theta^{-1}\right).  \label{asymptotics-t}
\end{equation}
Let us denote with $\Sigma _{\ST}$ the spectrum (the set of the eigenvalue functions $t(\lambda )$) of the transfer matrix $\ST(\lambda )$. By the definitions \rf{Tdef} and \rf{Lax}, then $\Sigma_{\ST}$ is contained\footnote{Here with $\mathbb{C}[x,x^{-1}]_{M}$ we are denoting the linear space of the Laurent polynomials of degree $M$ in the variable $x\in \mathbb{C}$.} in $\mathbb{C}[\lambda ^{2},\lambda^{-2}]_{(\SRN+\text{e}_{\SRN}-1)/2}$, where we have used the notation e$_{\SRN}=0$ for $\SRN$ odd and $1$ for $\SRN$ even.

Note that in the case of $\SRN$ even, the $\Theta $-charge naturally induces the grading $\Sigma_{\ST}=\bigcup_{k=0}^{l}\Sigma _{\ST}^{k}$, where:
\begin{equation}
\Sigma
_{\ST}^{k}\equiv \left\{ t(\lambda )\in \Sigma _{\ST}:\lim_{\log \lambda
\rightarrow \mp \infty }\lambda ^{\pm \SRN}t(\lambda )=\left( \prod_{a=1}^{\SRN}\frac{\kappa _{a}\xi _{a}^{\pm 1}}{i}\right) (q^{k}+q^{-k})\right\} .
\end{equation}
This simply follows by the asymptotics of $\ST(\lambda )$ and by its commutativity with $\Theta $. In particular,
any $t(\lambda )\in \Sigma_{\ST}^{k}$ is a $\ST$-eigenvalue corresponding to simultaneous eigenstates of
$\ST(\lambda )$ and $\Theta $\ with $\Theta $-eigenvalues $q^{\pm k}$.

\subsection{Cyclic SOV representations}

The separation of variables method of Sklyanin is based on the observation that the
spectral problem for $\ST(\la)$ simplifies considerably if one works in an auxiliary representation where the commutative family of operators $\SB(\la)$ is diagonal.

In the case of the Sine-Gordon model the vector space\footnote{It is always possible to provide the structure of Hilbert space to this finite-dimensional linear space. In particular, the scalar product in the SOV space is naturally introduced by the requirement that the transfer matrix is self-adjoint in the SOV representation. Appendix B addresses this issue.} $\BC^{p^\SRN}$ underlying the SOV representation can be identified with the space of functions $\Psi(\eta)$ defined for $\eta$ taken from the discrete set
\begin{equation}
{\mathbb B_\SRN}\,\equiv\,\big\{\,(q^{k_1}\zeta_1,\dots,q^{k_\SRN}\zeta_\SRN)\,;\,(k_1,\dots,k_\SRN)\in\BZ_p^\SRN\,\big\}\,,
\end{equation}
on these functions $\SB(\la)$ acts as a multiplication operator,
\begin{equation}\label{Bdef}
\SB_\SRN(\la)\,\Psi(\eta)\,=\,\eta_\SRN^{{\rm e}_\SRN}\,b_\eta(\la)\,\Psi(\eta)\,,\qquad b_\eta(\la)\,\equiv\,
\prod_{n=1}^{\SRN}\frac{\kappa _{n}}{i}\prod_{a=1}^{[\SRN]}\left( \la/\eta_a-\eta_a/\la\right)\,;
\end{equation}
where $[\SRN]\equiv\SRN-\en$ and $\eta_1,^{}\dots,\eta_{[\SRN]}^{}$ are the zeros of $b_\eta(\la)$. In the case of even $\SRN$ it turns out that we need a supplementary variable $\eta_\SRN$ in order to be able to parameterize the spectrum of $\SB(\la)$.

In \cite{NT} we have proven that for general values of the parameters $\kappa$ and $\xi$ of the original representation it is possible to construct these SOV representations and moreover we have defined the map which fixes the  SOV parameter $\eta$ in terms of the parameters $\kappa$ and $\xi$.

In these SOV representations the spectral problem for $\ST(\la)$ is reduced to the following discrete system of Baxter-like equations in the wave-function $\Psi_t(\eta)=\langle\,\eta\,|\,t\,\rangle$ of a $\ST$-eigenstate $|\,t\,\rangle$:
\begin{equation}\label{SOVBax1}
t(\eta _{r})\Psi(\eta)\,=\,{\tt a}(\eta _{r})\ST_r^-\Psi(\eta)+{\tt d}(\eta _{r})\ST_r^+\Psi(\eta)\, \qquad \text{ \ }\forall r\in \{1,...,[\SRN]\},
\end{equation}
where $\ST_r^{\pm}$ are the operators defined by
\[
\ST_r^{\pm}\Psi(\eta_1,\dots,\eta_\SRN)=\Psi(\eta_1,\dots,q^{\pm 1}\eta_r,\dots,\eta_\SRN)\,,
\]
while the coefficients ${\tt a}(\lambda )$ and ${\tt d}(\lambda )$ are defined by:
\begin{equation}
{\tt a}(\lambda )=\prod_{n=1}^{\SRN}\frac{\kappa _{n}}{i\lambda
_{n}}(1-iq^{-1/2}\lambda _{n}\kappa _{n})(1-iq^{-1/2}\frac{\lambda _{n}}{\kappa
_{n}}), \qquad {\tt d}(\lambda )=q^{\SRN}{\tt a}(-\lambda q).
\end{equation}%
In the case of $\SRN$ even we have to add to the system \rf{SOVBax1} the following equation in the variable $\eta_\SRN$:
\begin{equation}
\ST_\SRN^+\Psi_{\pm k}(\eta)\,=\,q^{\pm k}\Psi_{\pm k}(\eta),
\end{equation}
for $t(\lambda )\in \Sigma _{\ST}^{k}\ \ $with $\ k\in \{0,...,l\}.$ Note that the cyclicity of these SOV representations is expressed by the identification of $( {\ST}^{\pm}_{j} )^{p}$ with the identity for any $\ j\in \{1,...,\SRN\}.$

\section{SOV characterization of $\ST$-eigenvalues}\label{Compatib}
Let us introduce the one parameter family $D(\la)$ of $p\times p$ matrix:
\begin{equation}\label{D-matrix}
D(\la) \equiv
\begin{pmatrix}
t(\la)   &-{\tt d}(\la)&   0        &\cdots & 0 & -{\tt a}(\la)\\
-{\tt a}(q\la)& t(q\la)&-{\tt d}(q\la)& 0     &\cdots & 0 \\
      0       & {\quad} \ddots      & &     &     &         \vdots   \\
  \vdots           &     &  \cdots    &  &       &     \vdots   \\
     \vdots         &     &   & \cdots &       &   \vdots     \\
     \vdots   &            &    &  &  \ddots{\qquad}     &   0 \\
 0&\ldots&0& -{\tt a}(q^{2l-1}\la)& t(q^{2l-1}\la) &
-{\tt d}(q^{2l-1}\la)\\
-{\tt d}(q^{2l}\la)   & 0      &\ldots      &     0  & -{\tt a}(q^{2l}\la)& t(q^{2l}\la)
\end{pmatrix}
\end{equation}
where for now $t(\la )$ is just an even Laurent polynomial of degree $\SRN+$e$_{\SRN}-1$ in $\la$.
\begin{lem}
\label{detD}The determinant $\det_{p}$$D$ is an even Laurent
polynomial of maximal degree $\SRN+$e$_{\SRN}-1$ in $\Lambda \equiv \lambda ^{p}$.
\end{lem}
\begin{proof}
Let us start observing that $D$$(\lambda q)$ is obtained by $D(\lambda )$ exchanging the first and\ $p$-th column and after the first and
\ $p$-th row, so that 
\begin{equation}
\det_{p}\text{$D$}(\lambda q)=\det_{p}\text{$D$}(\lambda )%
\text{ \ \ }\forall \lambda \in \mathbb{C},
\end{equation}
which implies that $\det_{p}D$ is function of $\Lambda $.
Let us develop the determinant:
\begin{eqnarray}\label{detD-exp}
\det_{p}D(\Lambda ) &=&\prod_{h=1}^{p}{\tt a}(\lambda
q^{h})+\prod_{h=1}^{p}{\tt a}(-\lambda q^{h})-q^{\SRN}{\tt a}(\lambda ){\tt a}(-\lambda
)\det_{2l-1}D_{(1,2l+1),(1,2l+1)}(\lambda )  \notag \\
&&-q^{\SRN}{\tt a}(\lambda q){\tt a}(-\lambda q)\det_{2l-1}D%
_{(1,2),(1,2)}(\lambda )+t(\lambda )\det_{2l}D_{1,1}(\lambda
)\text{,}
\end{eqnarray}
where $D_{(h,k),(h,k)}(\lambda )$ denotes the $(2l-1)\times (2l-1)$
sub-matrix of $D(\lambda )$ obtained removing the rows and columns $h$ and $k$ while $D_{h,k}(\lambda )$ denotes the $2l\times 2l$
sub-matrix of $D(\lambda )$ obtained removing the row $h$ and
column $k$. The interest toward this decomposition of $\det_{p}$$D%
(\Lambda )$ is due to the fact that the matrices $D%
_{(1,2),(1,2)}(\lambda )$, $D_{(1,2l+1),(1,2l+1)}(\lambda )$ and 
$D_{1,1}(\lambda )$ are \textit{tridiagonal} matrices. Following the same reasoning used in Lemma \ref{Ap1} to prove that $\det_{2l}$$D_{1,1}(\lambda )$ is an even function of $\lambda$ we can also show that this is true for  $\det_{2l-1}$$D_{(1,2),(1,2)}(\lambda )$ and $\det_{2l-1}$$D_{(1,2l+1),(1,2l+1)}(\lambda )$. From the parity of these functions the parity of $\det_{p}$$D(\Lambda)$ follows by using \rf{detD-exp}.

Being ${\tt a}(\lambda )$, ${\tt d}(\lambda )$ and $t(\lambda )$\ Laurent polynomial of
degree $\SRN$ in $\lambda $, in the case of $\SRN$\ even the statement of the
lemma is already proven; so we have just to show that: 
\begin{equation}
\lim_{\log \Lambda \rightarrow \mp \infty }\Lambda ^{\pm \SRN}\det_{p}\text{%
$D$}(\Lambda )=0
\end{equation}%
for $\SRN$ odd which follows observing that: 
\begin{equation}
\lim_{\log \Lambda \rightarrow \mp \infty }\Lambda ^{\pm \SRN}\det_{p}\text{%
$D$}(\Lambda )=i^{\pm pN}\prod_{n=1}^{\SRN}\kappa _{n}^{p}\xi _{n}^{\pm
p}\det_{p}\left\Vert q^{-(1\mp 1)\SRN/2}\delta _{h,k+1}-q^{(1\mp 1)\SRN/2}\delta
_{h,k-1}\right\Vert .
\end{equation}
\end{proof}
The interest toward the function $\det_{p}D(\Lambda )$ is due to
the fact that it allows the following characterization of the $\ST$-spectrum: 
\begin{lem}
\label{Charact-Sigma1}$\Sigma _{\ST}$ is the set of all the functions $%
t(\lambda )\in \mathbb{C}[\lambda ^{2},\lambda ^{-2}]_{(\SRN+\text{e}_{\SRN}-1)/2}$ which
satisfy the system of equations:%
\begin{equation}
\det_{p}\text{$D$}(\eta^{p} _{a})=0\text{ \ \ }\forall a\in \{1,...,[\SRN]\} \  \ and \  \ (\eta _{1},...,\eta _{[\SRN]})\in 
\mathbb{B}_{\SRN},  \label{compatibility}
\end{equation}%
plus in the case of $\SRN$ even:%
\begin{equation}
\lim_{\log \Lambda \rightarrow \mp \infty }\Lambda ^{\pm \SRN}\det_{p}\text{%
$D$}(\Lambda )=0.  \label{asymp-compatibility}
\end{equation}
\end{lem}

\begin{proof}
The requirement that the system of equations \rf{SOVBax1} admits a non-zero solution leads to the equations (\ref{compatibility}), while the equation (\ref{asymp-compatibility}) for even $\SRN$
simply follows by observing that:
\begin{align}
\lim_{\log \Lambda \rightarrow \mp \infty }\Lambda ^{\pm \SRN}\det_{p}\text{%
$D$}(\Lambda )& =\det_{p}\left\Vert q^{(1\mp 1)\SRN/2}\delta
_{i,j-1}+q^{-(1\mp 1)\SRN/2}\delta _{i,j+1}-(q^{k}+q^{-k})\delta
_{i,j}\right\Vert   \notag \\
& \times (-1)\prod_{n=1}^{\SRN}\left( i\kappa _{n}\xi _{n}^{\pm }\right) ^{p}\left.
=\right. 0.
\end{align}%
\end{proof}

Note that the above characterization of the $\ST$-spectrum $\Sigma _{\ST}$
requires as input the knowledge of $\mathbb{B}_{\SRN}$, i.e. the lattice of
zeros of the operator $B(\lambda )$. It is so interesting to notice that
this characterization has in fact a reformulation which is independent from
the knowledge of $\mathbb{B}_{\SRN}$. To explain this let us note that
Lemma \ref{detD} allows to introduce the following map:
\begin{equation}
\mathcal{D}_{p,\SRN}:t(\lambda )\in \mathbb{C}[\lambda ^{2},\lambda
^{-2}]_{(\SRN+\text{e}_{\SRN}-1)/2}\rightarrow \mathcal{D}_{p,\SRN}(t(\lambda ))\equiv \det_{p}
\text{$D$}(\Lambda )\in \mathbb{C}[\Lambda ^{2},\Lambda
^{-2}]_{(\SRN+\text{e}_{\SRN}-1)/2}.
\end{equation}
In terms of this map we can introduce a further characterization of the
spectrum of the transfer matrix $\ST(\lambda )$.

\begin{thm}
The spectrum $\Sigma _{\ST}$ of the transfer matrix $\ST(\lambda )$ coincides
with the kernel $\mathcal{N}_{\mathcal{D}_{p,\SRN}}\subset \mathbb{C}[\lambda
^{2},\lambda ^{-2}]_{(\SRN+\text{e}_{\SRN}-1)/2}$ of the map $\mathcal{D}_{p,\SRN}$.
\end{thm}
\begin{proof}
The inclusion $\mathcal{N}_{\mathcal{D}_{p,\SRN}}\subset \Sigma _{\ST}$ is
trivial thanks to Lemma \ref{Charact-Sigma1}, vice-versa if $t(\lambda )\in
\Sigma _{\ST}$ then the function $\det_{p}$$D$$(\Lambda )$ is zero in $%
\SRN+$e$_{\SRN}$ different values of $\Lambda ^{2}$ which thanks to
Lemma \ref{detD} implies $\det_{p}$$D$$(\Lambda )\equiv 0$, i.e. $\Sigma
_{\ST}\subset \mathcal{N}_{\mathcal{D}_{p,\SRN}}$.
\end{proof}
That is the set of eigenvalues of the transfer matrix $\ST(\lambda )$ is exactly characterized as the subset of $\mathbb{C}[\lambda^{2},\lambda ^{-2}]_{(\SRN+\text{e}_{\SRN}-1)/2}$ which contains all the solutions of the functional equation
$\det_{p}D(\Lambda )=0$. In the next section we will show that this functional equation is nothing else that the Baxter
equation.

{\bf Remark 1.}  \ Let us note that the same kind of functional equation $\det D(\Lambda )=0$ also appears in \cite{BR89, Ne02, Ne03}. There it recasts, in a compact form, the functional relations which result from the truncated fusions of transfer matrix eigenvalues. It is so relevant to point out that for the {\it BBS-model}\footnote{The BBS-model \cite{BS, Ba89, BBP90, Ba04} has been analyzed in the SOV approach in a series of works \cite{GIPS06, GIPST07, GIPST08}.} in the SOV representation the non-triviality condition of the solutions of the system of Baxter-like equations has been shown \cite{GIPS06} to be equivalent to the truncation identity in the fusion of transfer matrix eigenvalues.

\section{Baxter functional equation}\label{Bax-FE}

The main consequence of the previous analysis is that it naturally leads to the complete characterization of the transfer matrix spectrum in terms of polynomial solutions of the Baxter functional equation.

\begin{thm}\label{Derivation-Baxter-functional}
Let $t(\lambda )\in \Sigma _{\ST}$ then $t(\lambda )$\ defines uniquely up to
normalization a polynomial $Q_{t}(\lambda )$ that satisfies the Baxter functional equation: 
\begin{equation}
t(\lambda )Q_{t}(\lambda )={\tt a}(\lambda )Q_{t}(\lambda q^{-1})+{\tt d}(\lambda
)Q_{t}(\lambda q)\ \ \ \ \forall \lambda \in \mathbb{C}.  \label{tq-Baxter}
\end{equation}
\end{thm}

\begin{proof}
The fact that given a $t(\lambda )\in \mathbb{C}[\lambda ^{2},\lambda
^{-2}]_{(\SRN+\text{e}_{\SRN}-1)/2}$ there exists up to normalization at most one polynomial $%
Q_{t}(\lambda )$\ that satisfies the Baxter
functional equation has been proven in Lemma 2 of \cite{NT}.
So we have to prove only the existence of $Q_{t}(\lambda )\in \mathbb{C}[\lambda ]$. An interesting point about the proof given here is that it is a constructive proof.

Let us notice that the condition $t(\lambda )\in \Sigma _{\ST}\equiv 
\mathcal{N}_{\mathcal{D}_{p,\SRN}}$ implies that the $p\times p$ matrix $D(\lambda )$
has rank $2l$ for any $\lambda \in \mathbb{C}\backslash \{0\}$. Let us denote with
\begin{equation}\label{cofactor-def}
\textsc{C}_{i,j}(\lambda)=(-1)^{i+j}\det_{2l}D_{i,j}(\lambda)
\end{equation}
the $(i,j)$ {\it cofactor} of the matrix $D$$(\lambda)$; then the matrix formed out of these cofactors has rank $1$, i.e. all the vectors:
\begin{equation}
\text{\textsc{V}}_{i}(\lambda )\equiv (\text{\textsc{C}}_{i,1}(\lambda ),%
\text{\textsc{C}}_{i,2}(\lambda ),...,\text{\textsc{C}}_{i,2l+1}(\lambda
))^{\ST}\in \mathbb{C}^{p}\text{ \ \ }\forall i\in \{1,...,2l+1\}
\end{equation}
are proportional: 
\begin{equation}
\text{\textsc{V}}_{i}(\lambda )/\text{\textsc{C}}_{i,1}(\lambda )=\text{%
\textsc{V}}_{j}(\lambda )/\text{\textsc{C}}_{j,1}(\lambda )\text{\ \ \ \ }%
\forall i,j\in \{1,...,2l+1\},\text{ }\forall \lambda \in \mathbb{C}.
\label{covector-proport}
\end{equation} 
The proportionality (\ref{covector-proport}) of the eigenvectors \textsc{V}$_{i}(\lambda
)$ implies: 
\begin{equation}\label{proportionality}
\text{\textsc{C}}_{2,2}(\lambda )/\text{\textsc{C}}_{2,1}(\lambda )=\text{%
\textsc{C}}_{1,2}(\lambda )/\text{\textsc{C}}_{1,1}(\lambda )
\end{equation}
which, by using the property (\ref{cofactors-diagonal}), can be rewritten as:%
\begin{equation}
\text{\textsc{C}}_{1,1}(\lambda q)/\text{\textsc{C}}_{1,2l+1}(\lambda q)=%
\text{\textsc{C}}_{1,2}(\lambda )/\text{\textsc{C}}_{1,1}(\lambda ).
\label{Inter-step}
\end{equation}%
Moreover, the first element in the vectorial condition\ 
$D(\lambda )$\textsc{V}$_{1}(\lambda )=$\b{0} reads: 
\begin{equation}\label{Bax-eq}
t(\lambda )\text{\textsc{C}}_{1,1}(\lambda )={\tt a}(\lambda )\text{\textsc{C}}%
_{1,2l+1}(\lambda )+{\tt d}(\lambda )\text{\textsc{C}}_{1,2}(\lambda ).
\end{equation}%
Let us note that from the form of ${\tt a}(\lambda ),$ ${\tt d}(\lambda )$ and $%
t(\lambda )\in \Sigma _{\ST}$ it follows that all the cofactors are Laurent
polynomial of maximal degree\footnote{The $a_{i,j}$ and $b_{i,j}$ are non-negative integers
and $\la^{(i,j)}_{h}\neq 0$ for any $h\in \{1,...,4l\SRN-(a_{i,j}+b_{i,j})\}$.} $2l\SRN$ in
$\lambda$:
\begin{eqnarray}\label{general-cofactor}
\text{\textsc{C}}_{i,j}(\lambda ) &=&\text{\textsc{c}}_{i,j}\lambda
^{-2l\SRN+a_{i,j}}\prod_{h=1}^{4l\SRN-(a_{i,j}+b_{i,j})}(\lambda
_{h}^{(i,j)}-\lambda ).
\end{eqnarray}

In Lemma \ref{cofactors-zeros}, we show that the equations \rf{Inter-step} and (\ref{Bax-eq}) imply that if \textsc{C}$_{1,1}(\lambda )$ has a common zero with \textsc{C}$_{1,2}(\lambda )$ then this is also a zero of 
\textsc{C}$_{1,2l+1}(\lambda )$ and that the same statement holds exchanging \textsc{C}$_{1,2}(\lambda )$ with \textsc{C}$_{1,2l+1}(\lambda )$. So we can denote with \textsc{c}$_{1,1}\overline{\text{\textsc{C}}}_{1,1}(\lambda )$, \textsc{c}$_{1,2l+1}\overline{\text{\textsc{C}}}_{1,2l+1}(\lambda )$ and \textsc{c}$_{1,2}\overline{\text{\textsc{C}}}_{1,2}(\lambda )$ the polynomials of maximal degree $4l\SRN$ obtained simplifying the common factors in \textsc{C}$_{1,1}(\lambda )$, \textsc{C}$_{1,2l+1}(\lambda )$ and \textsc{C}$_{1,2}(\lambda )$. Then, by equation \rf{Inter-step}, they have to satisfy the relations:
\begin{equation}\label{proportionality-cofactor}
\overline{\text{\textsc{C}}}_{1,2l+1}(\lambda )=q^{\bar{\SRN}_{1,1}}\text{$\overline{\text{%
\textsc{C}}}$}_{1,1}(\lambda q^{-1}),\text{ \ \ $\overline{\text{\textsc{C}}}
$}_{1,2}(\lambda )=q^{-\bar{\SRN}_{1,1}}\text{$\overline{\text{\textsc{C}}}$}_{1,1}(\lambda q)%
\text{ \ and \ \textsc{c}}_{1,2l+1}=\varphi \text{\textsc{c}}_{1,1},
\end{equation}%
where $\varphi \equiv $\textsc{c}$_{1,1}/$\textsc{c}$_{1,2}$ and $\bar{\SRN}_{1,1}$ is the degree of the polynomial $\overline{\text{\textsc{C}}}_{1,1}(\lambda )$. So that equation
(\ref{Bax-eq}) assumes the form of a Baxter equation in the polynomial $\overline{%
\text{\textsc{C}}}_{1,1}(\lambda )$: 
\begin{equation}\label{deform-BAX}
t(\lambda )\text{$\overline{\text{\textsc{C}}}$}_{1,1}(\lambda )=\bar{{\tt a}}%
(\lambda )\text{$\overline{\text{\textsc{C}}}$}_{1,1}(\lambda q^{-1})+\bar{{\tt d}}%
(\lambda )\text{$\overline{\text{\textsc{C}}}$}_{1,1}(\lambda q),
\end{equation}%
with coefficients $\bar{{\tt a}}(\lambda )\equiv q^{\bar{\SRN}_{1,1}}\varphi {\tt a}(\lambda )$ and $\bar{{\tt d}}%
(\lambda )\equiv q^{-\bar{\SRN}_{1,1}}\varphi ^{-1}{\tt d}(\lambda )$. Note that the consistence of the above equation implies that $\varphi $ is a $p$-root of the unity. Indeed, denoting
with $\bar{D}(\Lambda )$ the matrix defined as in \rf{D-matrix} but with coefficients $%
\bar{{\tt a}}(\lambda )$ and $\bar{{\tt d}}(\lambda )$, equation \rf{deform-BAX} implies: 
\begin{equation}\label{detDbar}
0= \det_{p} \bar{D}(\Lambda )\equiv (\varphi ^{p}-1)\left( \prod_{h=1}^{p}{\tt a}(\lambda q^{h}
)-\varphi ^{-p}\prod_{h=1}^{p}{\tt a}(-\lambda q^{h})\right).
\end{equation}%
The expansion for $\det_{p} \bar{D}(\Lambda )$ in \rf{detDbar} is derived by using the expansion \rf{detD-exp} for $\det_{p}\bar{D}(\Lambda )$, the formulae\footnote{They follow from the \textit{tridiagonality} of these matrices and by using Lemma \ref{Det-tridiag}.}:
\begin{align}
& \left. \det_{2l}\overline{D}_{1,1}(\lambda )=\det_{2l}%
D_{1,1}(\lambda ),\right.  \\
& \left. \det_{2l-1}\overline{D}_{(1,2),(1,2)}(\lambda
)=\det_{2l-1}D_{(1,2),(1,2)}(\lambda ),\right.  \\
& \left. \det_{2l-1}\overline{D}_{(1,2l+1),(1,2l+1)}(\lambda
)=\det_{2l-1}D_{(1,2l+1),(1,2l+1)}(\lambda ),\right. \text{
\ }
\end{align}%
and the condition $t(\lambda )\in \Sigma _{T}$. Finally, if we define:
\begin{equation}
Q_{t}(\lambda )\equiv \lambda ^{a}\text{$\overline{\text{\textsc{C}}}$}_{1,1}(\lambda ),
\end{equation}%
where $q^{-a}=q^{\bar{\SRN}_{1,1}}\varphi$ with $a\in\{0,..,2l\}$, we get the statement of the theorem.
\end{proof}

{\bf Remark 2.} \ The previous theorem implies that for any $t(\lambda )\in \Sigma _{\ST}$ the polynomial solution $Q_{t}(\lambda )$ of the Baxter equation can be related to the determinant of a tridiagonal matrix of finite size $p-1$. Note that the spectrum of the Sine-Gordon model in the case of irrational coupling $\bar{\beta}^2$ should be deduced from $\beta^2=p'/p$ rational in the limit $\beta^2\rightarrow\bar{\beta}^2$. In particular, this implies that under this limit ($p\rightarrow + \infty$) the dimension of the representation diverges as well as the size of the tridiagonal matrix whose determinant is associated to the solution $Q_{t}(\lambda )$ of the Baxter equation. It is then relevant to point out that in the case of the quantum periodic Toda chain the solutions of the corresponding Baxter equation are expressed in terms of determinants of semi-infinite tridiagonal matrices \cite{GM,GP,KL99}. 
\vspace{0.2cm}

It is worth noticing that the set of polynomials $Q_{t}(\lambda )$, introduced in the previous theorem, admits a more precise characterization: 
\begin{thm}
Let $t(\lambda )\in \Sigma _{\ST}$ then $t(\lambda )$\ defines uniquely up to
normalization a polynomial solution $Q_{t}(\lambda )$\ of the Baxter functional equation \rf{tq-Baxter} of maximal degree $2l\SRN$.

In the case $\SRN$ odd,\ it results: 
\begin{equation}
Q_{t}(0)\equiv Q_{0}\neq 0,\text{ \ \ and \ }\lim_{\lambda \rightarrow
\infty }\lambda ^{-2l\SRN}Q_{t}(\lambda )\equiv Q_{2l\SRN}\neq 0.
\label{cond-asymp}
\end{equation}
In the case $\SRN$ even, the condition (\ref{cond-asymp}) selects $t(\lambda
)\in \Sigma _{\ST}^{0}$ while for $t(\lambda )\in \Sigma _{\ST}^{k}$ with $k\in
\{1,...,l\}$ we have the characterization $Q_{0}=Q_{2l\SRN}=0$ and: 
\begin{equation}
\lim_{\lambda \rightarrow 0}\frac{Q_{t}(\lambda q)}{Q_{t}(\lambda )}=q^{\pm
k},\text{ \ \ }\lim_{\lambda \rightarrow \infty }\frac{Q_{t}(\lambda q)}{%
Q_{t}(\lambda )}=q^{-(\SRN\pm k)}.  \label{cond-asymp+k}
\end{equation}
\end{thm}

\begin{proof}
Thanks to formula (\ref{cofactors-parity}), the cofactor $\textsc{C}_{1,1}(\lambda ) \in\mathbb{C}[\lambda ,\lambda^{-1}]_{2l\SRN}$ is even in $\lambda$ and so it admits the expansions:
\begin{eqnarray}
\text{\textsc{C}}_{1,1}(\lambda ) &=&\text{\textsc{c}}_{1,1}\lambda
^{-2l\SRN+2\tilde{a}_{1,1}}\prod_{i=1}^{2l\SRN-(\tilde{a}_{1,1}+\tilde{b}_{1,1})}(\lambda
_{i}^{(1,1)}-\lambda )(\lambda _{i}^{(1,1)}+\lambda ).
\end{eqnarray}
Let us note now that by using the properties (\ref{cofactors-diagonal}) and (\ref{cofactors-parity}), the relation \rf{proportionality} can be rewritten as: 
\begin{equation}
\text{\textsc{C}}_{1,1}(\lambda q)\text{\textsc{C}}_{1,1}(\lambda )=q^{\SRN}%
\text{\textsc{C}}_{1,2}(\lambda )\text{\textsc{C}}_{1,2}(-\lambda ).
\label{cofactor-equality}
\end{equation}
Using that and the general representation \rf{general-cofactor} for the cofactor $\text{\textsc{C}}_{1,2}(\lambda )$, we get:
\begin{equation}\label{1-rel}
a_{1,2}=2\tilde{a}_{1,1}\equiv 2a\text{, \ \ }b_{1,2}=2\tilde{b}_{1,1}\equiv 2b\text{, \ \ 
\textsc{c}}_{1,2}^{2}=\text{\textsc{c}}_{1,1}^{2}q^{-2(\SRN +b)}
\end{equation}
and: 
\begin{equation}
\left( \lambda _{i}^{(1,1)}\right) ^{2}=\left( \lambda _{i}^{(1,2)}\right)
^{2}\equiv \bar{\lambda}_{i}^{2},\text{ \ }\left( \lambda
_{i+2l\SRN-(a+b)}^{(1,2)}\right) ^{2}=\left( \bar{\lambda}_{i}/q\right) ^{2}
\label{cofactor-zeros}
\end{equation}
with $\bar{\lambda}_{i}\neq 0$ for any $i\in \{1,...,2l\SRN-(a+b)\}$ with $a$ and  $b \in \mathbb{Z}^{\geq 0}$. Note that the equation \rf{1-rel} and the fact that $\varphi \equiv $\textsc{c}$_{1,1}/$\textsc{c}$_{1,2}$ is a $p$-root of the unity imply $\varphi=q^{b+\SRN}$. Then we can write: 
\begin{eqnarray}
\text{\textsc{C}}_{1,1}(\lambda ) &=&\text{\textsc{c}}\lambda
^{-2l\SRN+2a}\prod_{i=1}^{2l\SRN-(a+b)}(\bar{\lambda}_{i}+\lambda )(\bar{\lambda}%
_{i}-\lambda ),  \label{cof-11} \\
\text{\textsc{C}}_{1,2}(\lambda ) &=&q^{a}\text{\textsc{c}}%
\lambda ^{-2l\SRN+2a}\prod_{i=1}^{2l\SRN-(a+b)}(\bar{\lambda}_{i}+\lambda
)((-1)^{H(x-i)}\bar{\lambda}_{i}-\lambda q),  \label{cof-12}
\end{eqnarray}
where \textsc{c}$\equiv $\textsc{c}$_{1,1}$ and $H(n)\equiv \{0\text{ \ \ for }n<0\text{, \ \ }1\text{ \ \ for }n\geq 0\}$ is the Heaviside step function. Here, $x$ is a non-negative integer which is fixed to zero thanks to formula \rf{proportionality-cofactor}. Then the solution $Q_{t}(\lambda )$ of the Baxter equation \rf{tq-Baxter} belongs to $\mathbb{C}[\lambda ]_{2l\SRN}$ and has the form:
\begin{equation}
Q_{t}(\lambda )\equiv \lambda ^{a}\prod_{i=1}^{2l\SRN-(a+b)}(\bar{\lambda}_{i}-\lambda ).
\label{Q_t-form}
\end{equation} 
Let us show now the remaining statements of the theorem concerning the asymptotics of $Q_{t}(\lambda )$. To this aim we compute the limits: 
\begin{align}
\lim_{\log \lambda \rightarrow \mp \infty }\lambda ^{\pm 2l\SRN}\text{\textsc{C}%
}_{1,1}(\lambda )& =\det_{2l}\left\Vert q^{-(1\mp 1)\SRN/2}\delta
_{i,j+1}+q^{(1\mp 1)\SRN/2}\delta _{i,j-1}-(q^{k}+q^{-k})\delta _{\text{e}%
_{\SRN},1}\delta _{i,j}\right\Vert _{i\neq 1,j\neq 1}  \notag \\
& \times \prod_{h=1}^{\SRN}(\frac{\kappa _{h}\xi _{h}^{\pm 1}}{i})^{2l}\left.
=\right. (\delta _{\text{e}_{\SRN},1}(1+(2l+1)\delta _{k,0})-1)\prod_{h=1}^{\SRN}(%
\frac{\kappa _{h}\xi _{h}^{\mp 1}}{i})^{2l},
\end{align}
 which imply: 
\begin{equation}
a=b=0,  \label{asymp-Lem}
\end{equation}
for $\SRN$ odd and $\SRN$ even with $t(\lambda )\in \Sigma _{\ST}^{0}$, i.e. the
condition (\ref{cond-asymp}). In the remaining cases, $\SRN$ even and $t(\lambda )\notin \Sigma _{\ST}^{0}$, the same formula implies: 
\begin{equation}
a\neq 0\text{, \ }b\neq 0,
\end{equation}
so that $Q_{0}=Q_{2l\SRN}=0$, while the asymptotics behaviors (\ref{cond-asymp+k}) simply
follow taking the asymptotics of the Baxter equation satisfied by $%
Q_{t}(\lambda )$.
\end{proof}
\section{$\SQ$-operator: Existence and characterization}\label{Q-op-def}

Let us denote with ${\bf \Sigma}_{t}$ the eigenspace of the transfer matrix $\ST(\lambda )$ corresponding to the eigenvalue $t(\lambda )\in \Sigma _{\ST}$,
then:

\begin{defn}
Let $\mathsf{Q}(\lambda )$ be the operator family defined by: 
\begin{equation}
\mathsf{Q}(\lambda )|t\rangle \equiv Q_{t}(\lambda )|t\rangle \text{ \ \ }%
\forall |t\rangle \in {\bf \Sigma}_{t}\text{ \ and \ }\forall t(\lambda
)\in \Sigma _{\ST},
\end{equation}
with $Q_{t}(\lambda )$ the element of $\mathbb{C}[\lambda ]_{2l\SRN}$
corresponding to $t(\lambda )\in \Sigma _{\ST}$ by the injection\ defined in
the previous theorem.
\end{defn}

Under the assumptions $\xi $ and $\kappa $ real or imaginary numbers, which
assure the self-adjointness of the transfer matrix $\ST(\lambda )$ for $%
\lambda \in \mathbb{R}$, the following theorem holds:

\begin{thm}
The operator family $\mathsf{Q}(\lambda )$ is a Baxter $\SQ$-operator:

\begin{description}
\item[(A)] $\mathsf{Q}(\lambda )$ satisfies with $\ST(\lambda )$ the commutation
relations: 
\begin{equation}
\lbrack \mathsf{Q}(\lambda ),\ST(\mu )]=[\mathsf{Q}(\lambda ),\mathsf{Q}(\mu
)]=0\text{ \ \ }\forall \lambda ,\mu \in \mathbb{C}\text{,}
\end{equation}
plus the Baxter equation: 
\begin{equation}
\ST(\lambda )\mathsf{Q}(\lambda )={\tt a}(\lambda )\mathsf{Q}(\lambda
q^{-1})+{\tt d}(\lambda )\mathsf{Q}(\lambda q)\text{ \ \ }\forall \lambda \in 
\mathbb{C}.
\end{equation}

\item[(B)] $\mathsf{Q}(\lambda )$ is a polynomial of degree $2l\SRN$ in $%
\lambda $: 
\begin{equation*}
\mathsf{Q}(\lambda )\equiv \sum_{n=0}^{2l\SRN}\mathsf{Q}_{n}\lambda ^{n},
\end{equation*}
with coefficients $\mathsf{Q}_{n}$ self-adjoint operators.

\item[(C)] In the case $\SRN$ odd,\ the operator $\mathsf{Q}_{2l\SRN}=$\textsf{id}
and $\mathsf{Q}_{0}$ is an invertible operator.

\item[(D)] In the case $\SRN$ even, $\mathsf{Q}(\lambda )$ commutes with the $%
\Theta $-charge and\ the operator $\mathsf{Q}_{2l\SRN}$ is the orthogonal
projection onto the $\Theta $-eigenspace with eigenvalue 1. $\mathsf{Q}_{0}$
has non-trivial kernel coinciding with the orthogonal complement to the $%
\Theta $-eigenspace with eigenvalue 1.
\end{description}
\end{thm}

\begin{proof}
Note that the self-adjointness of the transfer matrix $\ST(\lambda )$ implies
that $\mathsf{Q}(\lambda )$ is well defined, indeed its action is defined on
a basis. The property (A) is a trivial consequence of Definition 1. Note that the injectivity of the map $t(\lambda )\in
\Sigma _{\ST}$ $\rightarrow Q_{t}(\lambda )\in \mathbb{C}[\lambda ]_{2l\SRN}$
implies: 
\begin{equation}
\left( Q_{t}(\lambda )\right) ^{\ast }=Q_{t}(\lambda ^{\ast })\text{ \ }%
\forall \lambda \in \mathbb{C}
\end{equation}
being $\left( {\tt a}(\lambda )\right) ^{\ast }={\tt d}(\lambda ^{\ast })$ and $\left(
t(\lambda )\right) ^{\ast }=t(\lambda ^{\ast })$. So we get the Hermitian
conjugation property $\left( \mathsf{Q}(\lambda )\right) ^{\dag }=\mathsf{Q}%
(\lambda ^{\ast })$, i.e. the self-adjointness of the operators $\mathsf{Q}%
_{n}$. The properties (C) and (D) of the operators $\mathsf{Q}_{0}$ and $%
\mathsf{Q}_{2l\SRN}$ directly follow from the asymptotics of the eigenfunction $%
Q_{t}(\lambda )$ while the commutativity of $\mathsf{Q}(\lambda )$ and $%
\Theta $ is a direct consequence of the commutativity of $\ST(\lambda )$ and $%
\Theta $.
\end{proof}

\section{Conclusion}

In the previous section we have shown that by only using the
characterization of the spectrum of the transfer matrix obtained by the SOV
method we were able to reconstruct the $\SQ$-operator. It is also interesting to
point out as the results derived in \cite{NT} together with those of the
present article yield:

\begin{thm}
The family $\CQ$ which characterizes the quantum integrability of the lattice
Sine-Gordon model (see definition (\ref{def-integrability})) is described by
the transfer matrix $\ST(\lambda )$ for a chain with $\SRN$ odd number of sites
while by $\ST(\lambda )$ plus the $\Theta $-charge for a chain with $\SRN$ even
number of sites.
\end{thm}
\begin{proof}
Let us start noticing that Proposition 3 and Theorem 4 of \cite{NT} are
derived only using the SOV method (i.e. without any assumption about the
existence of the $\SQ$-operator). So only using SOV analysis we have derived
that for $\SRN$ odd the transfer matrix $\ST(\lambda )$ has simple spectrum while
for $\SRN$ even this is true for $\ST(\lambda )$ plus the $\Theta $-charge; i.e.
they define a complete family of commuting observables and so satisfy the
properties (A) and (C) of the definition (\ref{def-integrability}). In this
article we have moreover shown that the $\SQ$-operator is defined as a function of
the transfer matrix which implies the property (B) of (\ref%
{def-integrability}) recalling that in \cite{NT} the time-evolution operator
$\SU$ has been expressed as a function of the $\SQ$-operator.
\end{proof}

\bigskip
\bigskip 

Let us shortly point out the main features required in abstract to extend to cyclic representations of other
integrable quantum models the same kind of spectrum characterization derived here for the lattice Sine-Gordon
model.

\begin{enumerate}
\item[{\bf R1.}] The model admits an SOV description and the spectrum of the transfer
matrix can be characterized by a system of Baxter-like equations in the
$\ST$-wave-function $\Psi(\eta)=\langle\,\eta\,|\,t\,\rangle$:
\begin{equation}\label{general-Baxter-system}
t(\eta _{r})\Psi(\eta)\,=\,{\tt a}(\eta _{r})\Psi(\eta_1,\dots,q^{-1}\eta_r,\dots,\eta_\SRN)+{\tt d}(\eta _{r})\Psi(\eta_1,\dots,q\eta_r,\dots,\eta_\SRN)\,,
\end{equation}
where $(\eta _{1},...,\eta _{\SRN})\in \mathbb{B}_{\SRN}$ with $\mathbb{B}_{\SRN}$ the set of zeros of the $B$-operator in the SOV representation. Here, the parameter $q$ is a root of unity defined as in (\ref{Weyl}) and
(\ref{beta}).
\end{enumerate}

Note that for cyclic representations of an integrable quantum model the set $\mathbb{B}_{\SRN}$ is a finite
subset of $\mathbb{C}^{\SRN}$. So the coefficients ${\tt a}(\eta _{r})$ and ${\tt d}(\eta _{r})$ are specified only in a finite number of points where they satisfy the following average value relations\footnote{The equations in \rf{average1} are trivial consequences of the SOV representation and of the cyclicity.}:
\begin{equation}\label{average1}
\CA(\eta_r^p)\,=\,\prod_{k=1}^{p}{\tt a}(q^k\eta_r)\,,\qquad
\CD(\eta_r^p)\,=\,\prod_{k=1}^{p}{\tt d}(q^k\eta_r)\,.
\end{equation}
Here $\CA(\Lambda)$ and $\CD(\Lambda)$ are the average values of the operator entries $A(\lambda)$ and $D(\lambda)$ of the monodromy matrix. Let us recall that the operator entries of the monodromy matrix are expected to be polynomials (or Laurent polynomials) in the spectral parameter $\lambda$ so the corresponding average values are polynomials (or Laurent polynomials) in $\Lambda \equiv \lambda^p$. It is then natural to introduce the functions ${\tt a}(\lambda)$ and ${\tt d}(\lambda)$ as polynomial (or Laurent polynomial) solutions of the following average relations:
\begin{equation}\label{a-d-functions}
\CA(\Lambda)+\gamma\CB(\Lambda)\,=\,\prod_{k=1}^{p}{\tt a}(q^k\lambda)\,,\qquad
\CD(\Lambda)+\delta\CB(\Lambda)\,=\,\prod_{k=1}^{p}{\tt d}(q^k\lambda)\,,
\end{equation}
where $\CB(\Lambda)$ is the average value of the operator $B(\lambda)$ and $\gamma$ and $\delta$ are constant to be fixed.

\begin{enumerate}
\item[{\bf R2.}] Let us denote with $Z_{f(\la)}$ the set of the zeros of the functions $f(\la)$, then:
\begin{equation}\label{R2}
\exists \ \lambda_{0} \ \in \ Z_{{\tt a}(\la)}  \ : \  \lambda_{0} \ \notin \ \cup_{h=0}^{2l-1} \ Z_{{\tt d}(\lambda q^{h})}.
\end{equation}
\end{enumerate}

\ {\bf R3.} The average values of the functions ${\tt a}$ and ${\tt d}$ are
not coinciding in all the zeros of the $B$-operator:
\begin{equation}
\CA(\eta _{a}^p)\neq \CD(\eta _{a}^p)\text{\ \ \ }\forall a\in \{1,...,[\SRN]\} \  \ \text{and} \  \ (\eta _{1},...,
\eta _{[\SRN]})\in \mathbb{B}_{\SRN}.
\end{equation}

The requirement {\bf R1} yields the introduction of the $p\times p$ matrix $D(\lambda )$, defined as in \rf{D-matrix}, by the functions ${\tt a}(\lambda)$ and ${\tt d}(\lambda)$ solutions of \rf{a-d-functions}. This should allow us to reformulate the spectral problem for the transfer matrix as the problem to classify all the solutions $t(\lambda )$ to the functional equation $\det_{p}D(\Lambda )=0$ in a model dependent class of functions.

The requirement {\bf R2} implies that the rank of the matrix $D(\lambda )$ is almost everywhere $2l$. Indeed, the condition \rf{R2} implies \textsc{C}$_{1,p}(\lambda_0)\neq 0$, independently from the function $t(\lambda )$. Being the cofactor \textsc{C}$_{1,p}(\lambda)$ a continuous function of the spectral parameter the above statement on the rank of the matrix $D(\lambda )$ follows. Under this condition we can follow the procedure presented in Theorem \ref{Derivation-Baxter-functional} to construct the solutions of the Baxter equation. Then the self-adjointness of the transfer matrix $\ST$ allows us to proceed as in section \ref{Q-op-def} to show the existence of the $\SQ$-operator as a function of $\ST$.

The requirement {\bf R3} is a sufficient criterion\footnote{
It is worth noticing that in the case of the Sine-Gordon model the
criterion {\bf R3} does not apply to the representations with $%
u_{n}=v_{n}=1$. Nevertheless, we have shown the simplicity of $\ST$ by using some model dependent properties of the coefficients ${\tt a}(\lambda)$ and ${\tt d}(\lambda)$, see section 5 of \cite{NT}.} to show the simplicity of the spectrum of $\ST$
which should imply that the full integrable structure of the quantum model
should be described by the transfer matrix as soon as the property (B) in
definition (\ref{def-integrability}) is shown for the model under
consideration.

Following the schema here presented, in a future publication we will address the analysis of the spectrum for the
so-called $\alpha$-{\it sectors} of the Sine-Gordon model (see \cite{NT}). The use of this approach is in particular relevant in these sectors of the Sine-Gordon model because a direct construction of the $\SQ$-operator leads to some technical difficulty.

\appendix

\section{Properties of the cofactors \textsc{C}$_{i,j}(\protect\lambda )$}

Let us consider an $M \times M$ {\it tridiagonal} matrix \footnote{An interesting analysis of the eigenvalue problem for tridiagonal matrices is presented in \cite{Sk05}.} $O$:
\begin{equation}
O \equiv
\begin{pmatrix}
z_1   &y_1&   0        &\cdots & 0 & 0\\
x_1& z_2&y_2& 0     &\cdots & 0 \\
     0       &  x_2      &  z_3 &    y_3  &     &         \vdots   \\
  \vdots          &     &    & \ddots &    &     \vdots   \\
   \vdots   &           &    &  &  \ddots{\qquad}     &   0 \\
 0&\ldots&0& x_{M-2} & z_{M-1} & y_{M-1}\\
0   & 0      &\ldots      &     0  & x_{M-1}& z_{M}
\end{pmatrix}
\end{equation}i.e. a matrix with non-zero entries only along the principal diagonal and the next upper and lower diagonals.

\begin{lem}
\label{Det-tridiag}The determinant of a tridiagonal matrix is invariant
under the transformation $\varrho _{\alpha }$ which multiplies for $\alpha $
the entries above the diagonal and for $\alpha ^{-1}$ the entries below the
diagonal leaving the entries on the diagonal unchanged.
\end{lem}

\begin{proof}
Let us note that the determinant of a tridiagonal matrix admits the
following expansion:
\begin{equation}
\det_{M}O=z_{1}\det_{M-1}O_{1,1}+x_{1}y_{1}\det_{M-2}O_{(1,2),(1,2)},
\end{equation}
where we have used the same notations introduced after formula \rf{detD-exp}. By using it,
we get that the action of $\varrho _{\alpha }$ reads:
\begin{equation}
\det_{M}\varrho _{\alpha }(O)=z_{1}\det_{M-1}\varrho _{\alpha
}(O)_{1,1}+x_{1}y_{1}\det_{M-2}\varrho _{\alpha }(O)_{(1,2),(1,2)}.
\end{equation}
Then the statement follows by induction noticing that the transformation $\varrho _{\alpha }$\ leaves always unchanged the determinant of a $2\times 2$ matrix.
\end{proof}
\begin{lem}
\label{Ap1}The following properties hold: 
\begin{equation}
\text{\textsc{C}}_{h+i,k+i}(\lambda )=\text{\textsc{C}}_{h,k}(\lambda q^{i})%
\text{ \ \ \ }\forall i,h,k\in \{1,...,2l+1\}\text{,}
\label{cofactors-diagonal}
\end{equation}
and: 
\begin{equation}
\text{\textsc{C}}_{1,1}(\lambda )=\text{\textsc{C}}_{1,1}(-\lambda )\text{ \
and \ \textsc{C}}_{2,1}(\lambda )=q^{\SRN}\text{\textsc{C}}_{1,2}(-\lambda ).
\label{cofactors-parity-0}
\end{equation}
\end{lem}

\begin{proof}
Note that by the definition \rf{cofactor-def} of the cofactors $\textsc{C}_{i,j}(\lambda)$ the equations (\ref{cofactors-diagonal}) are simple consequences of $q^{p}=1$\ and are proven exchanging rows and columns in the determinants.

Let us prove now that the cofactor \textsc{C}$_{1,1}(\lambda )=\det_{2l}$%
$D_{1,1}(\lambda )$ is an even function of $\lambda $. The
tridiagonality of the matrix $D_{1,1}(\lambda )$ allows us to use
the previous lemma: 
\begin{align}
\text{\textsc{C}}_{1,1}(\lambda )& \equiv \det_{2l}\left\Vert t(\lambda
q^{h})\delta _{h,k}-{\tt a}(\lambda q^{h})\delta _{h,k+1}-q^{\SRN}{\tt a}(-\lambda
q^{h+1})\delta _{h,k-1}\right\Vert _{h>1,k>1}  \notag \\
& =\det_{2l}\left\Vert t(\lambda q^{h})\delta _{h,k}-q^{\SRN}{\tt a}(\lambda
q^{h})\delta _{h,k+1}-{\tt a}(-\lambda q^{h+1})\delta _{h,k-1}\right\Vert
_{h>1,k>1}  \notag \\
& =\det_{2l}\left\Vert t(\lambda q^{h})\delta _{h,k}-{\tt d}(-\lambda q^{k})\delta
_{k,h-1}-{\tt a}(-\lambda q^{k})\delta _{k,h+1}\right\Vert _{h>1,k>1}  \notag \\
& \equiv \det_{2l}\left( D_{1,1}(-\lambda )\right) ^{T}=%
\text{\textsc{C}}_{1,1}(-\lambda ).  \label{C11_even}
\end{align}
To prove now the second relation in (\ref{cofactors-parity-0}) we expand the
cofactors:
\begin{eqnarray}
\text{\textsc{C}}_{2,1}(\lambda ) &=&\prod_{h=2}^{2l+1}{\tt a}(\lambda
q^{h})+{\tt d}(\lambda )\det_{2l-1}D_{(1,2),(1,2)}(\lambda ),
\label{C_2,1-expan} \\
\text{\textsc{C}}_{1,2}(\lambda ) &=&\prod_{h=1}^{2l}{\tt d}(\lambda
q^{h})+{\tt a}(\lambda q)\det_{2l-1}D_{(1,2),(1,2)}(\lambda )\text{%
.}  \label{C_1,2-expan}
\end{eqnarray}
By using the same steps shown in (\ref{C11_even}), the tridiagonality of the
matrix \textsc{D}$_{(1,2),(1,2)}(\lambda )$ implies that its determinant is
an even function of $\lambda $ from which the statement \textsc{C}$%
_{2,1}(\lambda )=q^{\SRN}$\textsc{C}$_{1,2}(-\lambda )$ follows recalling that $%
{\tt d}(\lambda )=q^{\SRN}{\tt a}(-\lambda q)$.
\end{proof}

\textbf{Remark 3.} \ In this article we need only the properties (\ref
{cofactors-parity-0}); however, it is worth pointing out that they are
special cases of the following properties of the cofactors:
\begin{equation}
\text{\textsc{C}}_{i,j}(\lambda )=q^{\SRN(i-j)}\text{\textsc{C}}%
_{j,i}(-\lambda )\text{ \ \ \ }\forall i,j\in \{1,...,2l+1\}.
\label{cofactors-parity}
\end{equation}

The proof of (\ref{cofactors-parity}) can be done similarly to that of (\ref
{cofactors-parity-0}) but we omit it for simplicity.
\vspace{0.2cm}

Let us use once again the notation $Z_{f}$ for the set of the zeros of a function 
$f(\lambda)$, then:
\begin{lem}\label{cofactors-zeros}
The equations (\ref{Inter-step}) and (\ref{Bax-eq}) imply:
\begin{equation}
\text{$Z$}_{\text{\textsc{C}}_{1,1}}\cap \text{$Z$}_{\text{\textsc{C}}%
_{1,2}}\equiv \text{$Z$}_{\text{\textsc{C}}_{1,1}}\cap \text{$Z$}_{\text{%
\textsc{C}}_{1,2l+1}}\text{.}
\end{equation}
\end{lem}
\begin{proof}
The inclusions $\left( Z_{\text{\textsc{C%
}}_{1,1}}\cap Z_{\text{\textsc{C}}_{1,2}}\right) \backslash Z_{\tt a}\subset Z_{%
\text{\textsc{C}}_{1,1}}\cap Z_{\text{\textsc{C}}_{1,2l+1}}$ and $\left( Z_{%
\text{\textsc{C}}_{1,1}}\cap Z_{\text{\textsc{C}}_{1,2l+1}}\right)
\backslash Z_{\tt d}\subset Z_{\text{\textsc{C}}_{1,1}}\cap Z_{\text{\textsc{C}}%
_{1,2}}$ trivially follow by equation (\ref{Bax-eq}). 

Let us observe now that \textsc{C}$_{1,2}(\lambda q^{-1})$ has no common
zero with ${\tt a}(\lambda )$ and that \textsc{C}$_{1,2l+1}(\lambda q)$ has no
common zero with ${\tt d}(\lambda )$. These statements simply follow from (\ref{C_1,2-expan}), (\ref{cofactors-diagonal})and(\ref{C_2,1-expan}) when we recall that ${\tt a}(\lambda)$ has
no common zero with $\prod_{h=0}^{2l-1}{\tt d}(\lambda q^{h})$ and that ${\tt d}(\lambda)$ has no common zero with $\prod_{h=2}^{2l+1}{\tt a}(\lambda q^{h})$. So, if $%
\left( Z_{\text{\textsc{C}}_{1,1}}\cap Z_{\text{\textsc{C}}_{1,2}}\right)
\cap Z_{\tt a}$ is not empty and $\lambda _{0}\in \left( Z_{\text{\textsc{C}}%
_{1,1}}\cap Z_{\text{\textsc{C}}_{1,2}}\right) \cap Z_{\tt a}$, the equation
(\ref{Inter-step}) computed in $\lambda =q^{-1}\lambda _{0}$ implies \textsc{%
C}$_{1,2l+1}(\lambda _{0})=0$ being \textsc{C}$_{1,2}(\lambda
_{0}q^{-1})\neq 0$, i.e. $\lambda _{0}\in Z_{\text{\textsc{C}}_{1,1}}\cap Z_{%
\text{\textsc{C}}_{1,2l+1}}$. Similarly, if $\left( Z_{\text{\textsc{C}}_{1,1}}\cap
Z_{\text{\textsc{C}}_{1,2l+1}}\right) \cap Z_{\tt d}$ is not empty and $\lambda _{0}\in \left(
Z_{\text{\textsc{C}}_{1,1}}\cap Z_{\text{\textsc{C}}_{1,2l+1}}\right) \cap Z_{\tt d}$, the equation (\ref{Inter-step}) computed in $\lambda =\lambda _{0}$ implies \textsc{C}$_{1,2}(\lambda _{0})=0$ being
\textsc{C}$_{1,2l+1}(\lambda _{0}q)\neq 0$, i.e. $\lambda _{0}\in Z_{\text{\textsc{C}}_{1,1}}\cap
Z_{\text{\textsc{C}}_{1,2}}$. So that (\ref{Inter-step}) implies the inclusions $\left(
Z_{\text{\textsc{C}}_{1,1}}\cap Z_{\text{\textsc{C}}_{1,2}}\right) \cap Z_{\tt a}\subset
Z_{\text{\textsc{C}}_{1,1}}\cap Z_{\text{\textsc{C}}_{1,2l+1}}$ and $\left( Z_{\text{\textsc{C}}_{1,1}}\cap
Z_{\text{\textsc{C}}_{1,2l+1}}\right) \cap Z_{\tt d}\subset Z_{\text{\textsc{C}}_{1,1}}\cap
Z_{\text{\textsc{C}}_{1,2}}$ in this way completing
the proof of the lemma.
\end{proof}

\section{Scalar product in the SOV space}
Here is described as a natural structure of Hilbert space can be provided to the linear space of the SOV representation by preserving the self-adjointness of the transfer matrix.

\subsection{Cyclic representations of the Weyl algebra}

Here, we consider the cyclic representations of the Weyl algebra $W_{q}^{(n)}$ in the case: 
\begin{equation}
\su_{n}^{p}=\sv_{n}^{p}=1\text{ for }\beta ^{2}=p^{\prime }/p\text{ with }%
p^{\prime }\text{ even and }p=2l+1\text{ odd.}
\end{equation}
At any site $n$ of the chain, we introduce the quantum space $\mathcal{R}%
_{n} $ with $\sv_{n}$-eigenbasis: 
\begin{equation}
\sv_{n}|k,n\rangle =q^{k}|k,n\rangle \text{ \ }\forall |k,n\rangle \in \text{B}%
_{n}=\{|k,n\rangle ,\forall k\in \{-l,...,l\}\}.  \label{v-eigenbasis}
\end{equation}
Note that the eigenvalues of $\sv_{n}$ describe the unit circle $\mathbb{S}%
_{p}=\{q^{k}:k\in \{-l,...,l\}\}$, indeed $q^{l+1}=q^{-l}$. On $\mathcal{R}%
_{n}$ is defined a $p$-dimensional representation of the Weyl algebra by
setting: 
\begin{equation}
\su_{n}|k,n\rangle =|k+1,n\rangle \text{ \ }\forall k\in \{-l,...,l\}
\end{equation}
with the cyclicity condition: 
\begin{equation}
|k+p,n\rangle =|k,n\rangle .
\end{equation}

\subsection{Representation in the SOV basis}

The analysis developed in \cite{NT} define recursively the
eigenbasis $\{|\bar{\eta}_{1}q^{h_{1}},...,\bar{\eta}_{\SRN}q^{h_{\SRN}}\rangle \}$ of the $B$-operator in the original representation, i.e. as linear combinations of the elements of the basis $\{|h_{1},...,h_{\SRN}%
\rangle \equiv\,\bigotimes_{n=1}^{\SRN}|h_{n},n\rangle \}$, where $|h_{n},n\rangle$ are the elements of the $\sv_{n}$-eigenbasis defined in (\ref{v-eigenbasis}). To write this
change of basis in a matrix form let us introduce the following notations: 
\begin{equation}
|y_{j}\rangle \equiv\,|\bar{\eta}_{1}q^{h_{1}},...,\bar{\eta}%
_{\SRN}q^{h_{\SRN}}\rangle \text{ \ and \ }|x_{j}\rangle
\equiv\,|h_{1},...,h_{\SRN}\rangle \text{\ }
\end{equation}
where: 
\begin{equation}
j:=h_{1}+\sum_{a=2}^{\SRN}(2l+1)^{(a-1)}(h_{a}-1)\in \{1,...,(2l+1)^{\SRN}\},
\label{corrisp}
\end{equation}
note that this defines a one to one correspondence between $\SRN$-tuples $%
\left( h_{1},...,h_{\SRN}\right) \in \{1,...,2l+1\}^{\SRN}$ and integers $j\in
\{1,...,(2l+1)^{\SRN}\}$, which just amounts to chose an ordering in the
elements of the two basis. Under this notation, we have: 
\begin{equation}
|y_{j}\rangle =\SW|x_{j}\rangle =\sum_{i=1}^{(2l+1)^{\SRN}}\SW_{i,j}|x_{i}\rangle ,
\end{equation}
where we are representing $|x_{j}\rangle $ as the vector $|j\rangle $ in the
natural basis in $\mathbb{C}^{(2l+1)^{\SRN}}$ and $\SW=||\SW_{i,j}||$ is a $%
(2l+1)^{\SRN}\times (2l+1)^{\SRN}$ matrix. The matrix $\SW$ is defined by recursion
in terms of the kernel $K$ constructed in appendix C of \cite{NT}, let us use the notation: 
\begin{equation}
K_{(\{h_{1},...,h_{\SRN}\},k_{1},\{k_{2},...,k_{\SRN}\})}\equiv\,K_\srn^{}(\,\eta\,|\,\chi_{\2}^{};\chi_\1^{}\,),
\end{equation}
where we are considering the case $\SRN-\SRM=1$. Then the recursion reads: 
\begin{equation}
\SW_{i,j}^{(\SRN)}=\sum_{k_{2},...,k_{\SRN}=1}^{2l+1}K_{(%
\{h_{1}(j),...,h_{\SRN}(j)\},h_{1}(i),\{k_{2},...,k_{\SRN}\})}\SW_{\bar{h}%
(i),a(k_{2},...,k_{\SRN})}^{(\SRN-1)},
\end{equation}
where we have introduced the index $(\SRN)$ and $(\SRN-1)$ in the matrices $\SW$ to
make clear the step of the recursion. Here, $(h_{1}(j)$,...,$h_{\SRN}(j))$ is
the unique $\SRN$-tuples corresponding to the integer $j\in \{1$,...,$%
(2l+1)^{\SRN}\}$ and $h_{1}(i)$ is the first entry in the unique $\SRN$-tuples
corresponding to the integer $i\in \{1$,...,$(2l+1)^{\SRN}\}$. Moreover, we
have defined: 
\begin{equation}
\bar{h}(i):=1+\frac{i-h_{1}(i)}{2l+1}\in \{1,...,(2l+1)^{(\SRN-1)}\}\text{ \ \
and \ \ }a(k_{2},...,k_{\SRN})=k_{2}+\sum_{a=3}^{\SRN}(2l+1)^{(a-2)}(k_{a}-1),
\end{equation}
\textbf{Remarks:}

a) Under the change of basis $\{|x_{j}\rangle \}\rightarrow \{|y_{j}\rangle
\}$ the generic operator X transforms for similarity: 
\begin{equation}
\text{X}_{SOV}\equiv\,\SW^{-1}\text{X}\SW,
\end{equation}%
so from the action of the zero operators $\eta _{a}$ and the shift operators
$\ST_{a}^{\pm }$ on the $B$-eigenbasis $|y_{j}\rangle $: 
\begin{equation}
\eta _{a}|y_{j}\rangle =\bar{\eta}_{a}q^{h_{a}(j)}|y_{j}\rangle \text{ \ and
\ } \ST_{a}^{\pm }|y_{j}\rangle =|y_{j\pm (2l+1)^{(a-1)}}\rangle 
\end{equation}%
we have that: 
\begin{equation}
\left( \eta _{a}\right) _{SOV}=\bar{\eta}_{a}||q^{h_{a}(j)}\delta _{i,j}||%
\text{ \ and \ }\left( \text{T}_{a}^{\pm }\right) _{SOV}=||\delta _{i,j\pm
(2l+1)^{(a-1)}}||\text{.}
\end{equation}%
From the above expression we have\footnote{%
Here, we are using the standard notation for the adjoint $X^{\dagger
}\equiv\,(X^{\ast })^{t}$.}: 
\begin{equation}
\left( \eta _{a}\right) _{SOV}^{\dagger }=\left( \eta _{a}\right)
_{SOV}^{\ast }\text{\ and \ }\left( \text{T}_{a}^{\pm }\right)
_{SOV}^{\dagger }=\left( \text{T}_{a}^{\mp }\right) _{SOV}\text{.}
\end{equation}

b) The known transformation properties of the entries of the monodromy
matrix in the original representation imply: 
\begin{equation}
\left( 
\begin{array}{cc}
\SD_{SOV}(\lambda ) & \SC_{SOV}(\lambda ) \\ 
\SB_{SOV}(\lambda ) & \SA_{SOV}(\lambda )%
\end{array}
\right) =\left( 
\begin{array}{cc}
\SG^{-1}\left( \SA_{SOV}(\lambda ^{\ast })\right) ^{\dagger }\SG & -\SG^{-1}\left(
\SB_{SOV}(\lambda ^{\ast })\right) ^{\dagger }\SG \\ 
-\SG^{-1}\left( \SC_{SOV}(\lambda ^{\ast })\right) ^{\dagger }\SG & \SG^{-1}\left(
\SD_{SOV}(\lambda ^{\ast })\right) ^{\dagger }\SG%
\end{array}
\right) ,  \label{sov-adj}
\end{equation}
with $\SG$ is a positive self-adjoint matrix defined by $\SG:=\SW^{\dagger }\SW$.

c) The quantum determinant relation is invariant under similarity
transformations and so we have: 
\begin{equation}
a(\lambda )d(\lambda q^{-1})=\SA_{SOV}(\lambda )\SD_{SOV}(\lambda
q^{-1})-\SB_{SOV}(\lambda )\SC_{SOV}(\lambda q^{-1}),  \label{Q-det-Sov}
\end{equation}

\begin{lem}
The basis $\{|y_{j}\rangle \}$ is not an orthogonal basis w.r.t. the natural
scalar product on $\{|x_{j}\rangle \}$.
\end{lem}

\begin{proof}

Note that the condition $\{|y_{j}\rangle \}$ is an orthogonal basis is
equivalent to the statement $\SG$ is a diagonal matrix (with positive diagonal
entries). Let us recall that the Hermitian conjugation property of $%
B(\lambda )$ together with the Yang-Baxter commutation relations imply: 
\begin{equation}
\lbrack \SB^{\dagger }(\lambda ),\SB(\mu )]=[\SB(\mu ),\SC(\lambda ^{\ast })]=\frac{%
q-q^{-1}}{\lambda ^{\ast }/\mu -\mu /\lambda ^{\ast }}(\SA(\lambda ^{\ast
})\SD(\mu )-\SA(\mu )\SD(\lambda ^{\ast }))\neq 0
\end{equation}%
that is the operator $B(\lambda )$ is not a normal operator. Now let us show
that the non-normality of $B(\lambda )$ implies that $\SG$ is not diagonal.
Indeed, we can write: 
\begin{align}
\lbrack \SB^{\dagger }(\lambda ),\SB(\mu )]& =\left( \SW^{\dagger }\right)
^{-1}\left( \SB_{SOV}(\lambda )\right) ^{\dagger }\SG \SB_{SOV}(\mu
)\SW^{-1}-\SW \SB_{SOV}(\mu )\SG^{-1}\left( \SB_{SOV}(\lambda )\right) ^{\dagger
}\SW^{\dagger }  \notag \\
& =\SW(\SG^{-1}(\SB_{SOV}(\lambda ))^{\dagger }\SG \SB_{SOV}(\mu )-\SB_{SOV}(\mu
)\SG^{-1}\left( \SB_{SOV}(\lambda )\right) ^{\dagger }\SG)\SW^{-1}.
\end{align}%
Note now that if we assume $\SG$ diagonal, then $\SG$\ commutes both with $\SB_{SOV}(\lambda )$ and with $(\SB_{SOV}(\lambda ))^{\dagger }$, being all
diagonal matrices in the SOV representation, which implies the absurd $%
[\SB^{\dagger }(\lambda ),\SB(\mu )]=0$.
\end{proof}

\subsection{Scalar product in the SOV space}

The self-adjointness of the family $\ST(\lambda )$ implies that the transfer
matrix eigenstates are orthogonal under the original scalar product: 
\begin{equation}
\delta _{i,j}=(|t_{i}\rangle ,|t_{j}\rangle ),
\end{equation}%
we have chosen the orthonormal ones. Note that the above equation naturally
induces a scalar product in the SOV representation obtained under change of
basis: 
\begin{equation}
(|b\rangle ,|a\rangle )_{SOV}\equiv\,(\SG|b\rangle ,|a\rangle )  \label{sov-scalar}
\end{equation}%
that is a scalar product for which the adjoint of a vector $|a\rangle $ is
the natural adjoint times the matrix $\SG$: 
\begin{equation}
|b\rangle ^{\dagger _{SOV}}\equiv\,\langle b|\SG\text{ \ with \ }\langle b|=\left(
\left( |b\rangle \right) ^{t}\right) ^{\ast },
\end{equation}%
and so for the generic operator $X$ we have: 
\begin{equation}
X^{\dagger _{SOV}}\equiv\,\SG^{-1}X^{\dagger }\SG.
\end{equation}%
It is trivial to notice that:

\begin{lem}
The family of operators $\ST_{SOV}(\lambda )$ is self-adjoint w.r.t. $\dagger
_{SOV}$ and the eigenstates $|t_{j}\rangle _{SOV}\equiv\,\SW^{-1}|t_{j}\rangle $ are
orthonormal w.r.t. the scalar product defined in (\ref{sov-scalar}). Moreover,
it results: 
\begin{equation}
\left( 
\begin{array}{cc}
\left( \SA_{SOV}(\lambda ^{\ast })\right) ^{\dagger _{SOV}} & \left(
\SB_{SOV}(\lambda ^{\ast })\right) ^{\dagger _{SOV}} \\ 
\left( \SC_{SOV}(\lambda ^{\ast })\right) ^{\dagger _{SOV}} & \left(
\SD_{SOV}(\lambda ^{\ast })\right) ^{\dagger _{SOV}}%
\end{array}
\right) =\left( 
\begin{array}{cc}
\SD_{SOV}(\lambda ) & -\SC_{SOV}(\lambda ) \\ 
-\SB_{SOV}(\lambda ) & \SA_{SOV}(\lambda )%
\end{array}
\right) .
\end{equation}
\end{lem}

\end{document}